\documentclass[letterpaper, 10pt, conference,dvipsname]{ieeeconf}

\IEEEoverridecommandlockouts                              % This command is only needed if 
\usepackage{cite}
\let\labelindent\relax
\usepackage{enumitem}
\usepackage{times}
\usepackage{soul}
\usepackage{url}
\usepackage[hidelinks]{hyperref}
\usepackage[utf8]{inputenc}
\usepackage[small]{caption}
\usepackage{graphicx}
\usepackage{booktabs}
\usepackage{algorithm, algpseudocode}
\usepackage[switch]{lineno}
\usepackage{amsmath,bbold,amsthm}
\usepackage{xcolor,nicematrix}

\newcommand{\ones}{\mathbb 1}
\newcommand{\reals}{{\mathbb{R}}}

\newcommand{\mnorm}[1]{{\left\vert\kern-0.25ex\left\vert\kern-0.25ex\left\vert #1 
    \right\vert\kern-0.25ex\right\vert\kern-0.25ex\right\vert}}

\newcommand{\mc}{\mathcal}

\newcommand{\half}{\frac{1}{2}}

\newtheorem{definition}{Definition} 

\newtheorem{lemma}{Lemma}
\newtheorem{corollary}{Corollary}
\newtheorem{remark}{Remark}
\newtheorem{proposition}{Proposition}
\newtheorem{assumption}{Assumption}
\newtheorem{example}{Example}

\usepackage{diagbox}
\usepackage{times,bm}
\usepackage{booktabs}
\usepackage[hidelinks]{hyperref}
\allowdisplaybreaks[2]

\begin{document}
\title{A Coupled Optimization Framework for Correlated Equilibria in Normal-Form Games}
% Game-theoretical aircraft re-routing in unpredictable and congested airspaces
\author{Sarah. H.Q. Li,
Yue Yu, 
Florian D\"{o}rfler, 
John Lygeros%<-this % stops a space
\thanks{S.H.Q. Li, F. D\"{o}rfler, and J. Lygeros are with the Automatic Control Laboratory, ETH Z\"{u}rich, Physikstrasse 3, Z\"{u}rich, 8092, Switzerland (email: (sarahli@control.ee.ethz.ch, dorfler@ethz.ch, jlygeros@ethz.ch)). Y. Yu is with the Oden Institute for Computational Engineering and Sciences, The University of Texas at Austin, Austin, TX, 78712, USA (email:yueyu@utexas.edu)}%
}

\maketitle
\thispagestyle{plain}
\pagestyle{plain}
%Main body starts

\begin{abstract}
In competitive multi-player interactions, simultaneous optimality is a key requirement for establishing strategic equilibria. This property is explicit when the game-theoretic equilibrium is the simultaneously optimal solution of coupled optimization problems. However, no such optimization problems exist for the correlated equilibrium, a strategic equilibrium where the players can correlate their actions. We address the lack of a coupled optimization framework for the correlated equilibrium by introducing an \emph{unnormalized game}---an extension of normal-form games in which the player strategies are lifted to unnormalized measures over the joint actions. We show that the set of fully mixed generalized Nash equilibria of this unnormalized game is a subset of the correlated equilibrium of the normal-form game. Furthermore, we introduce an entropy regularization to the unnormalized game and prove that the entropy-regularized generalized Nash equilibrium is a sub-optimal correlated equilibrium of the normal form game where the degree of sub-optimality depends on the magnitude of regularization. We prove that the entropy-regularized unnormalized game has a closed-form solution, and empirically verify its computational efficacy at approximating the correlated equilibrium of normal-form games.   
% Our formulation enables the direct application of Nash-seeking learning dynamics to learn correlated equilibria, which we leverage to formulate a projected gradient play algorithm for seeking correlated equilibria. We compare the properties of the projected gradient play algorithm against the feasibility program for computing the correlated equilibrium polytope.  
\end{abstract}

\section{Introduction}
As autonomous and artificial intelligence-assisted technology become ubiquitous  in our daily lives, game theory has emerged as an important tool for modeling and analyzing the interactions between autonomous agents. Within a game, player interactions are at an equilibrium when their strategies are \emph{simultaneously optimal}: no player can achieve a better objective by unilaterally deviating from its current strategy.  
% This property is known as \emph{simultaneous optimality} and matches the preferred interaction outcome in large-scale cyber-physical systems such as supply chains~\cite{leng2005game}, power markets~\cite{amin2020converging}, and urban mobility\cite{li2019tolling}. 
For equilibria concepts such as the Nash equilibrium and the Stackelberg equilibrium, simultaneous optimality is an explicit property: these equilibria solve coupled problems within an optimization framework. The existence of such a framework has also enabled the development of gradient-based algorithms for computing game-theoretic equilibria, in particular in autonomy and artificial intelligence~\cite{li2019tolling,bach2020two,adkins2019game}. 

The correlated equilibrium is an extension of the Nash equilibrium to the joint action space. By utilizing a \emph{correlation device} that enables players to coordinate their actions, a correlated equilibrium is more effective than the Nash equilibrium at optimizing the social welfare, especially in competitive games with three or more players~\cite{aumann1987correlated}. In particular, such games arise naturally in
% In particular, this is achieved without sacrificing simultaneous optimality. 
urban mobility~\cite{yang2019multi,li2023adaptive}, robotics~\cite{li2022congestion}, and power markets~\cite{lee2003solving}.
Since correlated equilibria form a connected polytope~\cite{nau2004geometry}, fairness and other system-level metrics can be optimized to global optimality. 
% In contrast, Nash equilibria is a disconnecte
% Furthermore, the set of correlated equilibrium forms a connected polytope~\cite{nau2004geometry}, which provides an easier mechanism design challenges, where the system operator may wish to design the Nash equilibrium to satisfy additional fairness and system level metrics. 

Despite its advantages, the correlated equilibrium becomes exponentially more expensive to compute as the number of players and actions increase, and the lack of an optimization framework has made it difficult to apply scalable gradient-based algorithms for finding correlated equilibria. Presently, we pose and answer the following question: \emph{can we construct a coupled optimization problem whose optimal solution is the correlated equilibrium of a normal-form game?}

\noindent\textbf{Contributions}. 
Our contribution is threefold-fold. 
\begin{enumerate}
    \item We introduce unnormalized games: an extension of normal-form games in which the player strategies are unnormalized measures over the joint action space. We prove that a strictly positive generalized Nash equilibrium of the unnormalized game is a correlated equilibrium of the normal-form game.
    \item We formulate an entropy-regularized unnormalized game, and prove that its generalized Nash equilibria are sub-optimal correlated equilibria of the normal-form game. Furthermore, we compute the generalized Nash equilibrium in closed-form, and show that its degree of sub-optimality as a correlated equilibrium depends on the entropy regularization. 
    % \item For all normal-form games, we prove that a sub-optimal correlated equilibrium via the generalized Nash equilibrium of an entropy-regularized unnormalized game, which has a closed-form solution. The sub-optimality of the approximate correlated equilibrium depends on the regularization magnitude.
    \item We empirically verify that the generalized Nash equilibria of entropy-regularized unnormalized games are sub-optimal correlated equilibria of normal-form games. Furthermore, we empirically derive the relationship between the degree of sub-optimality vs the magnitude of the entropy regularization of these generalized Nash equilibria.
\end{enumerate}
% We derive its coupled KKT conditions that characterize the resulting generalized Nash equilibrium, and show that  being a fully mixed correlated equilibrium of the normal-form game is equivalent to being a generalized Nash equilibrium of the unnormalized game. We then add an entropy regularization to the unnormalized game and show the resulting game has a closed form solution that is an $\epsilon$-correlated equilibrium of the normal-form game. 

\noindent\textbf{Relevant research}.
First introduced in~\cite{aumann1987correlated}, the correlated equilibrium exists in both finite and infinite games, including games that possess no Nash equilibria~\cite{hart1989existence}.  A correlated equilibrium definition requires both a correlation device and the resulting probability distribution over the joint action space~\cite{dhillon1996perfect}.  
% In this manuscript, we ignore the selection of correlation devices and instead, focus on representing the correlated equilibrium as the resulting probability distribution on the joint strategy space. 
% Similar to Nash equilibrium~\cite{myerson1978refinements}
The correlated equilibrium has multiple definitions and formulations under different assumptions~\cite{forges1993five,bach2020two}. 
Extensions of correlated equilibrium include constrained correlated equilibrium~\cite{boufous2023constrained}, quantal correlated equilibrium~\cite{vcerny2022quantal}, extensive-form correlated equilibrium~\cite{celli2020no}, and coarse correlated equilibrium~\cite{borowski2014learning}. A correlated equilibrium's stability properties are analyzed in~\cite{kohlberg1986strategic,okada1981stability}.
Learning dynamics that converge to the correlated equilibrium include uncoupled no-regret learning dynamics~\cite{hart2000simple} and evolution dynamics~\cite{arifovic2019learning}. Gradient-based learning dynamics are not well explored. 
% Gradient-based learning dynamics that converge to a correlated equilibrium is not well-studied for ~\cite{hart2000simple} showed that correlated equilibria arise naturally from uncoupled no-regret learning dynamics and other learning dynamics including evolution dynamics~\cite{arifovic2019learning}. %, and gradient descent dynamics~\cite{borowski2014learning}.
\section{Equilibria concepts in normal form games}\label{sec:prelim}
We consider a normal-form game with $N$ players. Let $[A_i] (A_i \in \mathbb{N})$ denote the set of actions available to player $i$, and let $[A] = [A_1]\times\ldots\times [A_N] (A = \prod_{i\in[N]} A_i)$ denote the set of all joint actions available in the game. We denote player $i$'s action as $a_i \in [A_i]$, the action taken by player $i$'s opponents as $a_{-i}$, and every player's joint action as $a := (a_1,\ldots, a_N) \in [A]$. Under a joint action $a\in[A]$, player $i$ incurs a cost $\ell_i(a)$, where $\ell_i: [A]\mapsto \reals$ for all $i \in [N]$. 

We denote the $A_i-$dimensional  probability simplex over $[A_i]$ as $\Delta_i$, the joint probability simplex as $\Delta = \Delta_1\times \ldots \times\Delta_N$, and the $A$-dimensional probability simplex over $[A]$ as $\Delta_A$. Player $i$'s \textbf{strategy} $x_i \in \Delta_i$ is a probability distribution over the action set $[A_i]$. Under the strategy $x_i$, player $i$ selects an action $a_i$ with the probability $x_i(a_i)$ for all $a_i \in [A_i]$. The \textbf{joint strategy} $x := (x_1,\ldots, x_N) \in \prod_{i\in[N]} \Delta_i$ is the collection of all of the players' strategies. Let the opponent strategy, action space, and strategy space be respectively given by
\[x_{-i} = \prod_{j\neq i} x_i,  \ [A_{-i}] = \prod_{j\neq i} [A_j], \ \Delta_{-i} = \prod_{j\neq i} \Delta_j, \  \forall i \in [N].\] 
Under the joint strategy $x$, the expected cost for player $i$ is given by
\begin{equation}\label{eqn:cost_in_expectation}
    \mathbb{E}_{a\sim x}[\ell_i(a)] = \sum_{a_i \in [A_i]} x_i(a_i)\sum_{a_{-i}\in [A_{-i}]}x_{-i}(a_{-i})\ell_i(a_i, a_{-i}).
\end{equation} 
We use $\hat{\ell}_i:[A_i]\times \Delta_{-i}\mapsto\reals$ to denote player $i$'s expected cost for playing action $a_i$ conditioned on the other players playing the strategy $x_{-i}$:
\begin{equation}\label{eqn:cost_overloading}
    \hat{\ell}_i(a_i; x_{-i}) = \mathbb{E}[\ell_i(a_i, a_{-i}) \ | a_{j}\sim x_{j}, \ \forall j \neq i], \forall i \in [N].  
\end{equation}
Using the notation $\hat{\ell}_i(a_i; x_{-i})$, player $i$'s expected cost~\eqref{eqn:cost_in_expectation} when choosing strategy $x_i$ is given by $\mathbb{E}_{a\sim x}[\ell_i(a)]  = \sum_{a_i \in [A_i]} x_i(a_i)\hat{\ell}_i(a_i;  x_{-i})$ when the other players choose strategies $x_{-i}$.

Each player minimizes its expected cost $\mathbb{E}_{a\sim x}[\ell_i(a)]$ through unilateral changes in its own strategy $x_{i} \in \Delta_i$. 
At the joint strategy $x = (x_1,\ldots,x_N)$ and for each $i \in [N]$, if $x_i$ minimizes $\sum_{a_i \in [A_i]} x_i(a_i)\hat{\ell}_i(a_i;  x_{-i})$ simultaneously, $x$ is a Nash equilibrium.  
\begin{definition}[\textbf{Nash equilibrium}]\label{def:NE}
    The joint strategy $x^\star = (x^\star_1,\ldots, x^\star_N)\in \Delta$ is a Nash equilibrium if for each $i \in [N]$, $x_i^\star$ satisfies 
    \begin{equation}\label{eqn:ne}
        \sum_{a_i \in [A_i]}x_i^\star(a_i)\hat{\ell}_i(a_i; x^\star_{-i})\leq  \sum_{a_i \in [A_i]}x_i(a_i)\hat{\ell}_i(a_i; x^\star_{-i}), \ \forall \ x_i \in \Delta_i. 
    \end{equation}
\end{definition}
The set of Nash equilibria is equivalent to the set of KKT points of the following coupled linear program for all $i \in [N]$. In general, the set is disconnected~\cite{nau2004geometry}. 
\begin{equation}\label{eqn:ne_optimization}
\begin{aligned}
    \min_{x_i \in \Delta_i}& \sum_{a_i \in [A_i]}x_i(a_i)\hat{\ell}_i(a_i; x_{-i}), \\ 
    \mbox{ s.t. } &\sum_{a_i \in [A_i]}x_i(a_i) = 1, x_i(a_i) \geq 0, \ \forall a_i \in [A_i],
\end{aligned}
\end{equation}
The concept of Nash equilibrium extends the notion of single-player optimality to \emph{simultaneous optimality} under unilateral deviations in the players' strategies~\cite{nash1951non}. A Nash equilibrium strategy $(x_1, \ldots, x_N)$ ensures that, within player $i$'s own strategy space $\Delta_i$, the strategy $x_i$ minimizes player $i$'s expected cost when the other players play strategy $x_{-i}$. 
% Without additional assumptions on the costs $\{\ell_i\}_{i\in[N]}$, Nash equilibria are the asymptotic limit points of simultaneous gradient play~\tc{red}{CITE}. 

\noindent\textbf{Independent decision-making induces inequity}. The concept of Nash equilibrium implicitly assumes that the players make decisions independently---i.e.,  $x_i, x_j$ are independent probability distributions for all $i, j \in [N]$, $j \neq i$. While this assumption holds for game-theoretic models such as the Prisoner's Dilemma~\cite{hamburger1973n}, it fails to take advantage of the additional coordination structure that exists in large-scale cyber-physical systems. Furthermore, independent decision-making often induces inequity among players. 

% \begin{example}[Crop  matching]\label{ex:vehi_standoff}
%     Consider a two player agriculture supply chain in which a farmer and a distributor must jointly select a type of crop among $[A_1] = [A_2]$ types of crop before each agricultural cycle. If they select different crop types, then both players earn nothing from this supply chain. If they jointly select the crop $a_1 = a_2$, then the total profit is given by $P(a_1, a_2) \in \reals_+$, and the players' market power in the specific crop market determines the portion of $P(a_1, a_2)$ that they earn. 

%     In this game, the Nash equilibrium strategy either leads to the same crop being chosen repeatedly, or each player independently randomizing among many crops. Repeatedly choosing the same crop results in unfair profit division in the long term, especially when the players' market powers differ significantly between crop markets, and players independently randomizing among crops will lead to production cycles with mismatched crop with positive probability.
% \end{example}
% This game is a variant of the Battle of  the Sexes game, where all pure Nash equilibria result in unfair profit division among the players, and a mixed Nash equilibrium results in mismatched crops a percentage of the time. 
\begin{example}[Vehicle standoff]\label{ex:vehi_standoff2}
Consider a single-lane road with bi-directional traffic and an unexpected pothole on its right side. Vehicles can choose to veer left or right to pass each other. Two pure Nash equilibria are (left, right) and (right, left),  but the traffic direction that chooses the pothole side will constantly be at a disadvantage. A mixed Nash equilibrium can ensure that both traffic directions are equally likely to encounter unexpected potholes, but it also means that with positive probability, both directions' vehicles will choose the same roadside and stall traffic. 
\end{example}
In Example~\ref{ex:vehi_standoff2}, a more robust solution is to \emph{coordinate} both traffic directions to alternate between the two Nash equilibria (left, right) and (right, left), without choosing the joint action pairs (right, right) or (left, left).  By doing so, the vehicles are choosing to \emph{correlate} their strategies. 
\begin{definition}[Correlated strategy]\label{def:c_strategy}
The $A$-dimensional probability distribution $y\in \Delta_A$ is a correlated strategy if
$ y(a) \geq 0 $ denotes the probability of the joint action $a  = (a_1,\ldots, a_N)$ occurring, for all $a\in [A]$ and $\sum_{a \in [A]} y(a) = 1$ \cite{roughgarden2010algorithmic,nau2004geometry}.
\end{definition} 
To employ correlated strategies, the players must have the incentive and the means to coordinate. As illustrated in
% Examples~\ref{ex:vehi_standoff} and
Example~\ref{ex:vehi_standoff2}, one possible incentive may be to ensure greater equity among players, and a possible coordination method is a traffic operator. %to negotiate (crop matching) or 

Correlated strategies require a \emph{correlation device}~\cite{dhillon1996perfect} that coordinate actions among players in order to be implementable. Presently, we assume that such a correlation device exists for every correlated strategy satisfying Definition~\ref{def:c_strategy}, so that the players can accurately coordinate and realize every joint action $a \in [A]$~\cite{nau2004geometry}.

Every joint strategy induces a correlated strategy, but not every correlated strategy can be reduced to a joint strategy. Furthermore, all correlated strategies induced by joint strategies are rank one in their tensor form.
\begin{example}[Rank of correlated strategy tensors]\label{ex:rank_1}
Consider a two-player normal form game with finite action sets $[U]$ and $[V]$. We will cast the correlated strategy $y \in \Delta_{UV}$ to a matrix $Y \in\reals^{U\times V}$.
For a joint strategy $(x_U, x_V)\in \Delta_{U}\times\Delta_V$, the corresponding correlated strategy $Y$ is given by \[Y =  x_{U}x_{V}^{\top}.\]
Thus, all the joint strategies $x = (x_U, x_V)$ produce rank one correlated strategies in its matrix form. 

On the other hand, let $Y_0$ be any feasible correlated strategy, then the complete set of correlated strategies is given by $Y_0 + \mathcal{N} $ where $\mathcal{N}$ is defined as
\[
\mathcal{N} = \{Y \in \reals^{U\times V}_+ | \sum_{u \in [U]}\sum_{v \in [V]}Y(u,v) = \ones^\top Y \ones = 0\}, 
\]
From the constraint $\ones^\top Y \ones = 0$, matrices in $\mathcal{N}$ have a maximum rank of $\min\{U,V\} - 1$. 
% Therefore, the matrix form of any correlated strategy has maximum rank $\min\{U,V\} - 1$.
\end{example}
Example~\ref{ex:rank_1}'s tensor formulation of correlated strategies is extendable to the $N$-player setting: every joint strategy $(x_1,\ldots, x_N)$, where $x_i \in \Delta_i$ for all $i \in [N]$, induces a correlated strategy $\hat{y}$ given by \begin{equation}\label{eqn:joint_action_distribution_from_NE}
    \hat{y}(a_1,\ldots,a_N) = \prod_{i \in [N]}x_i(a_i), \ \forall (a_1,\ldots,a_N) \in [A].
\end{equation}
If we cast $\hat{y} \in \Delta_{A}$ to an $N$-dimensional tensor $Y \in \reals^{A_1\times\ldots\times A_N}$, we observe that $\hat{y}$ is again a rank one tensor. 

\textbf{Comparison of solution spaces $\Delta$ and $\Delta_A$}. The joint strategy's and correlated strategy's solution spaces differ significantly in size. A joint strategy is given by $N$ independent probability distributions $x_i \in \Delta_i$, and its overall dimension is $\sum_{i}A_i$. 
On the other hand, a correlated strategy has dimension $A = \prod_{i\in[N]} A_i$. When the number of players or the number of player actions increases, the joint strategy space  $\Delta$'s dimension scales linearly, while the correlated strategy space  $\Delta_{A}$'s dimension scales \emph{exponentially}.

\textbf{Player optimality in correlated strategy space $\Delta_A$}. If a correlated strategy is optimal for player $i$, the joint action $(a_i, a_{-i})$ is played only when no other action $\hat{a}_i \in [A_i]$ can be played with $a_{-i}$ in place of $a_i$ to improve player $i$'s expected cost $\ell_i$. This notion of optimality is the same as the Nash equilibrium~\eqref{eqn:ne}. However, unlike the Nash equilibrium, defining the optimality of an independent strategy is no longer sufficiently descriptive.
We formally define correlated equilibrium as below. 
% Instead, a correlated equilibrium occurs when each player's chosen action is an individual best response for the joint opponent action that the chosen action is played against.
\begin{definition}[\textbf{Correlated equilibrium~\cite{aumann1987correlated}}]\label{def:CE}
    The correlated strategy $y \in \Delta_A$ (Definition~\ref{def:c_strategy}) is a correlated equilibrium if for all $i \in [N]$ and $a_i, \hat{a}_i \in [A_i]$, 
    \begin{equation}\label{eqn:ce}
        \sum_{a_{-i}\in [A_{-i}]} \Big(\ell_i(a_i, a_{-i}) - \ell_i(\hat{a}_i, a_{-i})\Big) y(a_i, a_{-i}) \leq 0.
    \end{equation}
\end{definition}
Intuitively, condition~\eqref{eqn:ce} implies that player $i$ cannot independently swap action $a_i$ for $\hat{a}_i$ while the other players play $a_{-i}$ and achieve a lower expected cost. On the set of correlated strategies induced by joint strategies, the correlated equilibrium condition~\eqref{eqn:ce} is equivalent to the Nash equilibrium condition~\eqref{eqn:ne}.
\begin{lemma}\label{lem:ce_ne_equivalence}
Over the set of correlated strategies induced by joint strategies as in~\eqref{eqn:joint_action_distribution_from_NE}, the correlated equilibrium condition~\eqref{eqn:ce} is equivalent to the Nash equilibrium condition~\eqref{eqn:ne}. 
\end{lemma}
\begin{proof}
    Over the subset of correlated strategies induced by a joint  strategy~\eqref{eqn:joint_action_distribution_from_NE}, the correlated equilibria condition~\eqref{eqn:ce} is equivalent to
    \begin{equation}
        \begin{aligned}
       \sum_{a_{-i} \in [A_{-i}]} \Big(\ell_i(a_i, a_{-i}) - \ell_i(\hat{a}_i, a_{-i})\Big) \prod_{j\in [N]} x_j(a_j) \leq 0\\
        x_i(a_i)\sum_{a_{-i} \in [A_{-i}]} \Big(\ell_i(a_i, a_{-i}) - \ell_i(\hat{a}_i, a_{-i})\Big) \prod_{j\neq i} x_j(a_j) \leq 0 \\
         x_i(a_i)\Big(\hat{\ell}_i(a_i;x_{-i}) - \hat{\ell}_i(\hat{a}_i;x_{-i})\Big)  \leq 0, \label{eqn:pf_0}
         \end{aligned}
    \end{equation}
for all $(a_i,\hat{a}_i) \in [A_i]$, $i \in [N]$. When $x_i(a_i) > 0$,~\eqref{eqn:pf_0} implies that  $\hat{\ell}_i(a_i; x_{-i}) = \min_{\hat{a}_i \in[A_i]} \hat{\ell}(\hat{a}_i; x_{-i})$. When $x_i(a_i) = 0$, the probability simplex constraints on $x_i$ enforce the existence of $a'_i \in [A_i]$, such that $x_i(a'_i) > 0$. Applying~\eqref{eqn:pf_0} to $a'_i$, $\hat{\ell}_i(a'_i; x_{-i}) = \min_{\hat{a}_i \in[A_i]} \hat{\ell}_i(\hat{a}_i; x_{-i})$, such that the cost $\hat{\ell}_i$ at the original action $a_i$ must have a cost greater than or equal to $\hat{\ell}_i(a'_i; x_{-i})$: $\hat{\ell}_i(a'_i; x_{-i}) \leq  \hat{\ell}_i({a}_i; x_{-i})$. In summary, the following holds for all $a_i \in [A_i], \ i \in [N]$:
    \[\begin{cases}
        \hat{\ell}_i(a_i; x_{-i}) = \min_{\hat{a}_i \in[A_i]} \hat{\ell}_i(\hat{a}_i; x_{-i}) & x_i(a_i) > 0\\
        \min_{\hat{a}_i \in[A_i]} \hat{\ell}_i(\hat{a}_i; x_{-i}) \leq  \hat{\ell}_i({a}_i; x_{-i}) & x_i(a_i) = 0
    \end{cases}.\]
    We define $\lambda_i = \min_{a_i\in[A_i]}\hat{\ell}_i({a}_i; x_{-i}) $ and $\mu_i(a_i) = \hat{\ell}_i(a_i; x_{-i}) - \lambda_i$ for each player $i\in[N]$. By construction, $(x_i, \lambda_i, \mu_i)$ satisfy the KKT conditions of~\eqref{eqn:ne_optimization} when the other players use strategies $x_{-i}$ for all $i \in [N]$, given by
    \begin{equation}\label{eqn:ne_kkt}
    \begin{aligned}
        \hat{\ell}_i(a_i; x_{-i}) - \lambda_i - \mu_i(a_i) &= 0, \quad \forall a_i \in [A_i]\\
         \ x_i(a_i) \geq 0, \ \mu_i(a_i) \geq 0, \ \mu_i(a_i)x_i(a_i) &= 0, \quad \forall a_i \in [A_i],\\
         \sum_{a_i \in [A_i]} x_i(a_i) & = 1.
    \end{aligned}
    \end{equation}
    Since the optimization problem~\eqref{eqn:ne_optimization} is a linear program in $x_i$, the KKT conditions fully characterize optimality and $(x_1,\ldots,x_N)$ is a Nash equilibrium.
    
    To show that a joint strategy $x = (x_1,\ldots, x_N)$ produces a correlated equilibrium $y$~\eqref{eqn:joint_action_distribution_from_NE} if $x$ is a Nash equilibrium, we first note that $x$ is a Nash equilibrium if and only if there exist Lagrange multipliers $\lambda_i$ and $\mu_i$ for all $i\in [N]$ such that $(x_i, \lambda_i, \mu_i)$ satisfies the KKT conditions~\eqref{eqn:ne_kkt} against opponent strategy $x_{-i}$. Furthermore, the KKT condition implies that
    \begin{equation}\label{eqn:pf_1}
        \lambda_i = \min_{\hat{a}_i\in[A_i]} \hat{\ell}_i(\hat{a}_i; x_{-i}), \ \forall i \in [N].
    \end{equation} 
    When $x_i(a_i) > 0$, $\hat{\ell}_i(a_i; x_{-i})\leq \hat{\ell}_i(\hat{a}_i; x_{-i})$ for all $\hat{a}_i \in [A_i]$. When $x_i(a_i) = 0$, $x_i(a_i)\Big(\hat{\ell}_i(a_i; x_{-i})- \hat{\ell}_i(a_i; x_{-i})\Big) = 0$. We can conclude that~\eqref{eqn:pf_0} holds for all $i \in [N]$, such that the correlated strategy $y$ constructed via~\eqref{eqn:joint_action_distribution_from_NE} from the Nash equilibrium $(x_1,\ldots, x_N)$ satisfies~\eqref{eqn:ce}.
\end{proof}
Lemma~\ref{lem:ce_ne_equivalence} provides an optimization-based proof for previously known equivalence between Nash equilibria and correlated equilibria~\cite{aumann1974subjectivity,canovas1999nash,nau2004geometry}.

\textbf{Correlated equilibrium polytope}. In the original~\cite{aumann1987correlated} formulation of correlated equilibrium, the set of correlated equilibria is shown to be equivalent to the following linear polytope on the joint action space.
\begin{multline}\label{eqn:ce_polytope}
    \mc{P}_{CE} := \Big\{y \in \Delta \ | \ \ones^\top y = 0, y \geq 0\\
    \sum_{a_{-i} \in [A_{-i}]}y(a_i, a_{-i})\Big(\ell_i(a_i, a_{-i}) - \ell_i(\hat{a}_i, a_{-i})\Big) \leq 0, \\
    \forall a_i, \hat{a}_i \in [A_i], \ i \in [N]\Big\}.
\end{multline}
In~\cite{nau2004geometry}, the authors showed that in addition to being a connected polytope, $\mc{P}_{CE}$'s boundary set $\partial\mc{P}_{CE}$ contains the correlated strategies induced by Nash equilibria. However, computing the $\mc{P}_{CE}$ suffers from the curse of dimensionality both due to the dimension of $\Delta$ being exponential in the $N$ and the number of $\mc{P}_{CE}$'s constraints, $\sum_{i\in[N]}{A_i \choose 2}$, being  exponential in $A_i$. 
% In~\cite{jiang2015polynomial,papadimitriou2008computing}, the authors derived an 'Ellipsoid against hope' algorithm to find feasible points of $\mc{P}_{CE}$. 

% We can still define the expected cost for player $i$ to take action $a$
% \[\hat{\ell}_i(\cdot, a_{-i}) = \sum_{a_{-i} \in [A]^{N-1}}\ell_i(\cdot, a_{-i})y_{\cdot, a_{-i}} = \sum_{a_i, a_{-i} \in [A]^N} \delta_{a_i}(\cdot)\ell_i(a, a_{-i})y_{a, a_{-i}}, \]
% where $\delta_{a_i}:[A] \mapsto \{0, 1\}$ is 0 if input is not $a_i$, and $1$ otherwise. 

% \begin{definition}[\textbf{Perfect Correlated equilibrium distributions}~\cite{dhillon1996perfect}]\label{def:CE_perfect}
%     Consider a correlated equilibrium $y \in \Delta_A$, we say $y$ is a perfect correlated equilibrium distribution if there exists a sequence of joint strategy distributions $\{z^k\}_{k \in \mathbb{N}} \subset \Delta_A$ and a sequence of diminishing real numbers $\{\epsilon^k\}_{k \in \mathbb{N}} \subset \reals_++$ satisfying $\lim_{k\rightarrow\infty} \epsilon^k = 0$, such that for all $k\in\mathbb{N}$ and $i \in [N]$, $y^k = (1 - \epsilon^k) y + \epsilon^k z^k$ satisfies~\eqref{eqn:ce} for all $i \in [N]$.
% \end{definition}
% \subsection{Joint probability distribution equivalent of Nash equilibrium}

% \input{contents/ce_gne_equivalence_new_version}
% \input{contents/entropy_CE_new_version}
\section{Lifting correlated equilibrium to generalized Nash equilibrium}
While a correlated equilibrium has the interpretation of being a `simultaneously optimal' strategy in literature, this interpretation lacks an explicit optimization formulation like the one that exists for Nash equilibrium in the form of~\eqref{eqn:ne_optimization}. In this section, we formulate a novel game in which the player strategy spaces are lifted from the probability measure space over $[A_{i}]$ to the unnormalized measure space over $[A] = \prod_{i\in[N]} [A_i]$. We show that a fully mixed generalized Nash equilibrium of the lifted game over unnormalized measures corresponds to a fully mixed correlated equilibrium of the normal form game. 

\subsection{Unnormalized measures}
We first consider a relaxation of probability distributions to unnormalized measures with finite mass~\cite{swendsen2014unnormalized,frogner2015learning}. 
\begin{definition}[Unnormalized measure]\label{def:unnormal_strategy}
Given an action set $[A]$ and a corresponding probability simplex $\Delta_A$, the vector $\alpha \in \reals_+^{A}/\{0\}$ is an unnormalized measure over the sample space $[A]$ if $\alpha(a) \geq 0$ for all $a \in [A]$.  
\end{definition}
Given two unnormalized measures $\alpha_1, \alpha_2$  over $[A]$, we denote their \textbf{element-wise product} by $\alpha_1 \circ\alpha_2$, such that 
\[(\alpha_1 \circ\alpha_2) \in \reals_+^{A}, (\alpha_1 \circ\alpha_2)(a) = \alpha_1(a)\alpha_2(a), \ \forall a \in [A].\] 
We consider the decomposition of a correlated strategy $y$ (Definition~\ref{def:c_strategy})  into $N$ unnormalized measures. 
\begin{definition}[Normalized Decomposition]\label{def:decompose_ce} Given a correlated strategy $y \in \Delta_A$, we say that $(\alpha_1,\ldots,\alpha_N)$ is a normalized decomposition of $y$ and $y$ is a product of $(\alpha_1,\ldots\alpha_N)$ if 
\begin{equation}\label{eqn:decomposing_y}
   y = \alpha_1\circ\ldots\circ\alpha_N, \ \alpha_i \in \reals^A_+, \ \forall i \in [N]. 
\end{equation}
\end{definition}
The mapping $(\alpha_1,\ldots,\alpha_N)\mapsto y$ is surjective but not injective. Every correlated strategy has at least one valid decomposition, but multiple sets of unnormalized measures can produce the same correlated strategy. 
\begin{lemma}\label{lem:support_inv}
    Every correlated strategy $y \in \Delta_A$ has an infinite number of decompositions in the form of~\eqref{eqn:decomposing_y}. Furthermore, if $(\alpha_1,\ldots,\alpha_N)$ satisfies
    \begin{equation}\label{eqn:support_constraint}
        \ones^\top(\alpha_1 \circ \ldots\circ\alpha_N) = 1,   \alpha_i \geq 0, \ \forall i \in [N],
    \end{equation}
    then $y = \alpha_1 \circ \ldots\circ\alpha_N$ is a correlated strategy. 
\end{lemma}
\begin{proof}
    Consider the unnormalized measures $(\alpha_1,\ldots,\alpha_N)$ that satisfy~\eqref{eqn:support_constraint} and $y =\alpha_1\circ\ldots\circ\alpha_N$, then $y$ is a correlated strategy as defined in Definition~\ref{def:c_strategy}. To show that every correlated strategy $y$ has an infinite number of decomposition's, we first show there exist at least one feasible decomposition: $\alpha_1 = y$, $\alpha_j = \ones$ for all $j \neq 1$. We then select $d\in \reals$, $d > 0$ and define the measures $\hat\alpha_1 = d\alpha_1$, $\hat\alpha_2 = \frac{1}{d}\alpha_2$, and $\hat\alpha_{j} = \alpha_j$ for all $j\notin \{1,2\}$, then the product of $(\hat\alpha_1,\hat\alpha_2,\ldots\circ\hat\alpha_N)$ satisfies $y =\hat\alpha_1\circ\hat\alpha_2\ldots,\hat\alpha_N$. Since we can select arbitrary positive real number $d$, there exists an infinite number of decompositions.  
\end{proof}
\begin{example}[Normalized decompositions]\label{ex:unnnormalized_distributions}
In a two-player, finite action game where $A_1 = A_2 \in\mathbb{N}$. We can represent the unnormalized measures by $A_1\times A_2$-dimensional matrices, $\alpha \in \reals^{A_1\times A_2}$, such that any element-wise product $\alpha_i\circ\alpha_j$ is equivalent to the Hadarmard product between their matrix counterparts. The following are all valid normalized decompositions and their correlated strategy product.
\begin{equation}\label{eqn:decomposition_NE}
    \left\{\begin{pNiceMatrix}[last-row]
0 & 0 & 0 \\
1 & 1 & 1 \\
 0 & 0 & 0\\
       & \alpha_1    &
\end{pNiceMatrix},  \ \begin{pNiceMatrix}[last-row]
0 & 0 & 1 \\
0& 0& 1 \\
 0 & 0 & 1\\
       & \alpha_2    &
\end{pNiceMatrix}\right\}, \begin{pNiceMatrix}[last-row]
0 &0 & 0 \\
0& 0& 1 \\
 0 & 0 & 0\\
       & y_{\alpha}    &
\end{pNiceMatrix},
\end{equation}
\begin{equation}\label{eqn:decomposition_positive}
   \left\{\frac{1}{45}\begin{pNiceMatrix}[last-row]
1 & 1 & 1 \\
1 & 1 & 1 \\
1 & 1 & 1\\
       & \beta_1    &
\end{pNiceMatrix},  \ \begin{pNiceMatrix}[last-row]
1 & 2 & 3 \\
4& 5& 6 \\
 7 & 8 & 9\\
       & \beta_2    &
\end{pNiceMatrix}\right\}, \frac{1}{45}\begin{pNiceMatrix}[last-row]
1 & 2 & 3 \\
4& 5& 6 \\
 7 & 8 & 9\\
       & y_{\beta}    &
\end{pNiceMatrix},  
\end{equation}
\begin{equation}\label{eqn:decomposition_non_invariant}
    \left\{\begin{pNiceMatrix}[last-row]
0 & 0 & 0 \\
0 & 0 & 1 \\
0 & 1 & 0\\
       & \gamma_1    &
\end{pNiceMatrix},  \ \half\begin{pNiceMatrix}[last-row]
1 & 0 & 0 \\
0& 1& 1 \\
 0 & 1 & 1\\
       & \gamma_2    &
\end{pNiceMatrix}\right\},  \frac{1}{2}\begin{pNiceMatrix}[last-row]
0 &0 & 0 \\
0& 0& 1 \\
 0 & 1 & 0\\
       & y_{\gamma}    &
\end{pNiceMatrix}.
\end{equation}
\end{example}

\subsection{Unnormalized game} 
We define an \textbf{unnormalized game} as the following extension of the normal-form game: instead of choosing probability distributions supported on each player's individual action space, each player $i$ chooses an unnormalized measure $\alpha_i$ over the joint action space $[A]$ as defined in Definition~\ref{def:unnormal_strategy}, constrained by the condition that the product $y = \alpha_1\circ\ldots\alpha_N$ is a correlated strategy. The player objectives remain the expected cost incurred by each player~\eqref{eqn:cost_overloading}, which is a multi-linear function of the unnormalized measures through~\eqref{eqn:decomposing_y} and~\eqref{eqn:cost_overloading}. Each player's optimization problem is given by
\begin{equation}\label{eqn:ce_individual_opt}
\begin{aligned}
    \min_{\alpha_i\in \reals_+^{A}} & \sum_{a \in [A]} {\ell}_i(a) \alpha_1(a)\ldots\alpha_N(a), \\
    \mbox{s.t. } &\sum_{a \in [A]}(\alpha_1\circ \ldots\circ\alpha_N)(a) = 1. \\ %\frac{1}{A_{-i}}
    % & \alpha_i(a) \geq 0, \ \forall a \in [A]. %\ones^\top \alpha = \ones^\top,
\end{aligned}
\end{equation}
In the unnormalized game, each player's strategy  $\alpha_i$ has the same dimension as the correlated strategy of the original finite game. Given the other players' strategies $\alpha_{-i}$, player $i$ uses their strategy $\alpha_i$ to optimize the expected cost $\sum_{a \in [A]} {\ell}_i(a) \alpha_1(a)\ldots\alpha_N(a)$, constrained by a mass constraint: $\sum_{a \in [A]}(\alpha_1\circ \ldots\circ\alpha_N)(a) = 1$. 
% We say that $\alpha_i$ is the best response to opponent strategies $\alpha_{-i}$ if $\alpha^i$ is in the argmin set of~\eqref{eqn:ce_individual_opt}.  
The optimal solution of this game is a generalized Nash equilibrium, where, in addition to minimizing their expected cost, each player's strategy must be feasible with respect to the other players' strategies. 
\begin{definition}[Generalized Nash equilibrium]\label{def:gNE}
    A joint strategy $(\alpha_1^\star,\ldots, \alpha_N^\star)$ is a generalized Nash equilibrium if for all $i \in [N]$,
    $\alpha^\star_i$ is the optimal solution to~\eqref{eqn:ce_individual_opt}.
\end{definition}
A generalized Nash equilibrium extends the standard Nash equilibrium (Definition~\ref{def:CE}) to games where each player's strategy feasibility depends on the other players' strategies. In the unnormalized game, all players share the strategy constraint given by~\eqref{eqn:support_constraint}. Next, we restrict our analysis to fully mixed unnormalized measures and show that when a set of fully mixed unnormalized measures forms a generalized Nash equilibrium, their product is a correlated equilibrium.  
\begin{assumption}[Fully mixed measures]\label{assum:support invariance} A measure $\alpha \in \reals_+^{A}$ is fully mixed if for all $a \in [A]$, $\alpha(a) > 0$. 
\end{assumption}
When a correlated strategy is fully mixed, all of its normalized decompositions $(\alpha_1,\ldots,\alpha_N)$ satisfy Assumption~\ref{assum:support invariance}. 
% Normalized decompositions that satisfy Assumption~\ref{assum:support invariance} are also deterministic: given $\alpha_{-i}$ and $y$, there exist a unique $\alpha_i$ such that $(\alpha_1,\ldots, \alpha_N)$ is a feasible decomposition of $y$. 
Games with certain player cost structures, such as zero-sum games and games with non-dominant strategies, tend to have fully mixed Nash and correlated equilibria~\cite{solan2002correlated,reny1999existence}. 
\subsection{Equivalence between generalized Nash equilibrium and correlated equilibrium}
For fully mixed correlated strategies $y$ with a normalized decomposition $(\alpha_1,\ldots,\alpha_N)$, we can show that $y$ is a correlated equilibrium of the normal form game if and only if $(\alpha_1,\ldots,\alpha_N)$ is a generalized Nash equilibrium of the unnormalized  game. 
\begin{proposition}\label{prop:ce_necessary}
    If $(\alpha_1,\ldots,\alpha_N)$ is a generalized Nash equilibrium of the unnormalized game~\eqref{eqn:ce_individual_opt}, and the $\alpha_{i}$'s satisfy Assumption~\ref{assum:support invariance}, then the product $y = \alpha_1\circ\ldots\circ\alpha_N$ is a correlated equilibria of the normal form game~\eqref{eqn:ce}.
\end{proposition}
\begin{proof}
We prove this proposition by showing that  Assumption~\ref{assum:support invariance} and the coupled KKT conditions of~\eqref{eqn:ce_individual_opt} together imply the correlated equilibrium condition in~\eqref{eqn:ce}.
From~\cite[Thm.3.3]{facchinei2010generalized}, the coupled KKT conditions of~\eqref{eqn:ce_individual_opt} are necessary and sufficient for $(\alpha_1,\ldots,\alpha_N)$ to be a generalized Nash equilibrium of the unnormalized game. Therefore, we show that~\eqref{eqn:ce} holds for $y = \alpha_1\circ\ldots\circ\alpha_N$ for all the KKT points $(\alpha_1,\ldots\alpha_N)$ of~\eqref{eqn:ce_individual_opt}. 

From the unnormalized game~\eqref{eqn:ce_individual_opt} for player $i$, we assign the Lagrange multipliers $\sigma_{i } \in \reals$ for the constraint $\sum_{a \in [A]}(\alpha_1\circ \ldots\circ\alpha_N)(a) = 1$ and $\mu_{i}(a)$ for the constraints $\alpha_{i}(a)\geq 0$. The first-order gradient condition and the complementarity condition of the KKT are given by 
    \begin{multline}
        \ell_i(a) \alpha_{-i}(a) - \sigma_{i}\alpha_{-i}(a) - \mu_{i}(a) = 0,  \\ \mu_{i}(a) = \begin{cases}
        \geq 0 & \alpha_{i}(a) = 0\\
        = 0 & \alpha_{i}(a) > 0
    \end{cases},\forall (i,a) \in [N]\times [A].
    \end{multline}
When $\alpha_{-i}(a) > 0$, the KKT conditions above imply that
\begin{equation}
    \ell_i(a) \begin{cases}
     = \sigma_{i}, & \text{if }\alpha_{i}(a) > 0 \\
    \geq \sigma_{i}, & \alpha_i(a) = 0 
\end{cases}, \ \forall (i,a) \in [N]\times[A].
\end{equation}
From Assumption~\ref{assum:support invariance}, $\alpha_{-i}(a) > 0$ and $\alpha_{i}(a) > 0$ for all $a \in [A]$. Therefore, $\mu_i(a) = 0$, $\sigma_{i} = \ell_i(a)$ for all $a \in [A]$. In particular, $\ell_i(a_i, a_{-i}) = \ell_i(\hat{a}_i, a_{-i})$ for all $a_i, \hat{a}_i \in [A_i]$. The correlated equilibrium condition~\eqref{eqn:ce} $(\ell_i(a_{i}, a_{-i}) - \ell_i(\hat{a}_i, a_{-i}))y(a_i, a_{-i})$ will then evaluate strictly to $0$ for all $\hat{a}_i \in [A_i]$ and $i \in [N]$.
We conclude that $y = \alpha_{1}\circ\ldots\circ\alpha_{N}$ is a correlated equilibrium.
\end{proof}
\begin{remark}
Proposition~\ref{prop:ce_necessary} suggests that a correlated equilibrium is fully mixed only if $\ell_{i}(a) $ evaluates to the same value for all $a\in[A]$. While this may seem restrictive, we use entropy regularizations in Section~\ref{sec:entropy_reg_equilibria} to produce games in which the regularized costs are all equal for each opponent action $a_{-i}$. We can show that the generalized Nash equilibrium under regularization will approximate the correlated equilibrium of the normal-form game even if no fully mixed correlated equilibrium exists. 
\end{remark}
The reverse of Proposition~\ref{prop:ce_necessary} is not true: if $y$ is a fully mixed correlated equilibrium, then it may not have a normalized decomposition that is a generalized Nash equilibrium of the unnormalized game. 
\begin{example}[Correlated equilibria not captured by gNE]
Consider a $2\times 2$ matrix game where player one chooses the row and player two chooses the column. The player costs are given by matrices $A$ and $B$, respectively, defined as 
\[P_1 = \begin{bmatrix}
    3& 3 \\ 2 &4
\end{bmatrix}, \ P_2 =  \begin{bmatrix}
    1& 2 \\ 1 &0
\end{bmatrix}.\]
We vectorize the joint action space as $[A] = \{a_1, a_2, a_3, a_4\}$, corresponding to the counterclockwise sequence of joint actions in matrix $P_i$ starting from the top left. For $y \in \Delta_4$ to be a correlated equilibrium as defined in~\eqref{eqn:ce}, it must satisfy
\begin{equation}\label{eqn:bimatrix_ce}
 \begin{aligned}
    3y(a_1) + 3 y(a_4) - 2y(a_2) - 4 y(a_3) &\leq 0\\
    1 y(a_1) + 1 y(a_2) - 0 y(a_3) - 2 y(a_4) &\leq 0
\end{aligned}   
\end{equation}
We can verify that $y_{CE} = \begin{bmatrix}
    \frac{1}{4} & \frac{1}{4} & \frac{1}{4} & \frac{1}{4}
\end{bmatrix}$ satisfies~\eqref{eqn:bimatrix_ce}. The unnormalized game played is given by
\begin{equation}\label{eqn:bimatrix_gne}
    \begin{aligned}
        \min_{\alpha_1} & \half \Big(3\alpha_1(a_1) + 2\alpha_1(a_2) + 4 \alpha_1(a_3) + 4\alpha_1(a_4)\Big) \\
        \mbox{s.t. } & \alpha_1(a_1) + \alpha_1(a_4) + \alpha_1(a_2) + \alpha_1(a_3) = 2\\
        & \alpha_1(a_i) \geq 0, \ \forall i \in [4], 
    \end{aligned}
\end{equation}
where the mass constraint simplifies since player two's strategy is $\alpha_2(a_j) = \half$ for all $j \in [4]$.
Consider the decomposition $\alpha_1 = \alpha_2 = \begin{bmatrix}\half & \half & \half & \half\end{bmatrix}$, we can verify that $\alpha_1$ does not minimize~\eqref{eqn:bimatrix_gne}. Specifically, $\hat{\alpha}_1 = \begin{bmatrix}
    0 & 1 & \half & \half
\end{bmatrix}$ can achieve a lower objective than $\alpha_1$ against $\alpha_2 = \begin{bmatrix}\half & \half & \half & \half\end{bmatrix}$. 
\end{example}
From Propositions~\ref{prop:ce_necessary}, we can say that if the generalized Nash equilibrium of~\eqref{eqn:ce_individual_opt} is strictly positive,  then their product $y$ is a fully mixed correlated equilibrium of the original normal-form game. A natural follow-up question remains: when do strictly positive correlated strategies exist? We explore this in the following section using entropy regularizations.

\section{Entropy-Regularized Correlated Equilibria}\label{sec:entropy_reg_equilibria}
A coupled optimization formulation for the correlated equilibrium expands the set of analysis techniques applicable to it. In this section, we demonstrate how entropy regularization can be applied to the unnormalized  game to find $\epsilon$-correlated equilibria of the original normal-form game. 

We consider the entropy-regularized counterpart of the unnormalized game~\eqref{eqn:ce_individual_opt}, where each player solves the optimization problem given by 
\begin{equation}\label{eqn:ce_entropy_reg}
\begin{aligned}
  \min_{\alpha_i\in \reals_+^{A}}& \sum_{a \in [A]} \Bigg({\ell}_i(a)+ \frac{1}{\lambda_i}\Big(\log\big({\alpha_i}(a)\big) - 1\Big)\Bigg) \alpha_i(a)\alpha_{-i}(a) \\
\mbox{s.t. } & \sum_{a \in [A]}(\alpha_1\circ \ldots\circ\alpha_N)(a) = 1. %\ones^\top \alpha = \ones^\top,  
\end{aligned}
\end{equation}
Here, $\lambda_i \geq 0$, $\lambda_i\in \reals$ denotes the magnitude of the entropy regularization. The total entropy of the correlated strategy $y = \alpha_1\circ\ldots\circ\alpha_{N}$ is given by $\sum_{a\in[A]} y(a)\log({y(a)})$, such that $\sum_{a\in[A]} y(a)\log({\alpha_i}) $ is equivalent to player $i$'s contribution to the total entropy.  For two unnormalized measures that achieve equal costs $\sum_{a \in [A]} {\ell}_i(a) \alpha_i(a) \alpha_{-i}(a)$,~\eqref{eqn:ce_entropy_reg} will favor the measure with the greater entropy and thus achieving a lower cost as defined by~\eqref{eqn:ce_individual_opt}.
% \begin{remark}
%     Since~\eqref{eqn:ce_entropy_reg} has a strongly convex objective and a player-wise convex constraint, the generalized Nash equilibrium of the game can be characterized by the optimal solutions of the coupled KKT conditions~\cite{dreves2011solution}. 
% \end{remark}
\begin{remark}
In applications of game-theoretic equilibrium, player costs and transitions are often obtained from noisy and imperfect data. When faced with such modeling inaccuracy data, it is in the player's best interest to seek strategies that not only optimize their expected cost but also maximize the entropy over the action space. 
Entropy regularization has been used in single-agent reinforcement learning and Nash equilibrium to find optimal policies that are robust to modeling inaccuracies and `unforeseen changes in the environment'~\cite{guan2020learning,savas2019entropy}. The Nash equilibrium of entropy regularized games is also known as \emph{logit quantal response equilibrium}~\cite{mckelvey1995quantal} and is an important tool for finding $\epsilon$-Nash equilibria in policy gradient-based reinforcement learning~\cite{cen2021fast}. Furthermore, the logit quantal response equilibrium is shown to be a more robust equilibrium model than the Nash equilibrium  for games models involving human players~\cite{alsaleh2022road}.
\end{remark}
We can show that the entropy-regularized unnormalized distribution game has the following closed-form solution.
% , which can be used in conjunction with Corollary~\ref{cor:approx_ce} to approximate the correlated equilibrium while incurring significantly less computation complexity. 
\begin{proposition}\label{prop:entropy_reg}
In the entropy-regularized unnormalized game where each player solves~\eqref{eqn:ce_entropy_reg}, there exists a generalized Nash equilibrium $(\alpha_1,\ldots,\alpha_N)$, where $\alpha_i$ is given by
    \begin{multline}\label{eqn:gNE_entropy_solution}
        \alpha_{i}(a) = \frac{\exp\big(-\lambda_i\ell_i(a)\big)}{\Big(\sum_a\exp\big(-\sum_j\lambda_j\ell_j(a)\big) \Big)^{\lambda_i/ \sum_{j}\lambda_j}}, \\ \forall i, a \in [N]\times[A].
    \end{multline}
\end{proposition}
\begin{proof}
Since~\eqref{eqn:ce_entropy_reg} has strongly convex objectives, convex independent constraints, and a shared constraint that is convex in each individual $\alpha_i$, the generalized Nash equilibrium of the game is equivalent to the coupled KKT points of~\eqref{eqn:ce_entropy_reg}~\cite[Thm.3.3]{facchinei2010generalized}. 
We assign the Lagrange multipliers $\sigma_i \in \reals$ to the constraint 
 $\sum_{a \in [A]}(\alpha_1\circ \ldots\circ\alpha_N)(a) = 1$ and $\mu_i(a) \in \reals_+$ to the constraints $\alpha_i(a) \geq 0$, for all $a \in [A]$. The KKT conditions of~\eqref{eqn:ce_entropy_reg} are given by 
\begin{equation}~\label{eqn:kkt_condition}
    \begin{aligned}
        0 =&  \Big(\ell_i(a) + \frac{1}{\lambda_i}\log\big(\alpha_i(a)\big)+\sigma_i\Big)\alpha_{-i}(a)- \mu_i(a), \ \forall a \in [A], \\
        1 = &\sum_{a \in [A]}(\alpha_1\circ \ldots\circ\alpha_N)(a),\\
        0 \leq & \alpha_i(a), 0 \leq \mu_i(a), 0 =\alpha_i(a) \mu_i(a),\ \forall a \in [A].
    \end{aligned}
\end{equation}
When $\alpha_{-i}(a) > 0$ and $\alpha_i(a) > 0$, the first KKT condition reduces to $0 = \ell_i(a) + \frac{1}{\lambda_i}\log\big(\alpha_i(a)\big) + \sigma_i$. We can solve for $\alpha_i(a)$ as
\begin{equation}\label{eqn:pf_5}
    \alpha_i(a) = \exp\big(-\lambda_i(\sigma_i+\ell_i(a)\big).
\end{equation}
For each joint action $a\in [A]$, the product $(\alpha_1\circ\ldots\circ\alpha_N)(a)$ is given by 
\begin{equation}\label{eqn:pf_6}
 (\alpha_1\circ\ldots\circ\alpha_N)(a) = \exp\Big(-\sum_{i}\lambda_i\sigma_i - \sum_i\lambda_i\ell_i(a)\Big).   
\end{equation}
From primal feasibility of the KKT conditions, $1 = \sum_{a \in [A]}(\alpha_1\circ \ldots\circ\alpha_N)(a)$. We combine this with~\eqref{eqn:pf_6} to derive
    \begin{equation}\label{eqn:pf_condition_2}
        \exp\big(-\sum_i \lambda_i\sigma_i\big) \sum_{a} \exp\big(-\sum_{i}\lambda_i\ell_i(a)\big) = 1.
    \end{equation}
    Let $\lambda_N = \sum_{i}\lambda_i$ and $\sigma_i = \sigma$ for all  $i \in [N]$. 
    We can then solve for $\exp(\sigma)$ in~\eqref{eqn:pf_condition_2} as
    $\Big(\sum_a \exp\big(-\sum_i\lambda_i \ell_i(a)\big)\Big)^{-1/\lambda_N}$.     Let $\beta_i(a) = \exp\big(-\lambda_i\ell_i(a)\big)$, each $\alpha_i(a)$~\eqref{eqn:pf_5} is given by 
    \begin{equation}
        \alpha_{i}(a) = \frac{\beta_i(a)}{\Big(\sum_a\prod_{j}\beta_j(a)\Big)^{\lambda_i/\lambda_N}}.
    \end{equation}
    When $\lambda_i, \ell_i(a)$ are finite, $\alpha_i(a) >0 $ for all $i, a \in [N]\times[A]$. Therefore $\alpha_1,\ldots,\alpha_N$ satisfy the KKT conditions and are the optimal solutions to the unnormalized game.
    % We will construct some Lagrange multipliers for $\alpha_1,\ldots\alpha_N$~\eqref{eqn:gNE_entropy_solution} and show that they satisfy the KKT conditions. 
\end{proof}
The correlated equilibrium corresponding to~\eqref{eqn:gNE_entropy_solution} is 
\begin{equation}\label{eqn:ce_entropy_solution}
    y(a) = \frac{\exp\big(-\sum_j\lambda_j\ell_j(a)\big)}{\sum_a\exp\big(-\sum_j\lambda_j\ell_j(a)\big)}, \ \forall a \in [A].
\end{equation}
The resulting correlated equilibrium is a softmax function over the regularized and weighted sum of individual player costs, where the level of entropy introduced is controlled by $\lambda_i$: the smaller the $\lambda_i$ is, the closer the resulting correlated equilibrium is to the completed mixed correlated strategy, $y'(a) = 1/A $ for all $a \in [A]$. 
We note that while Proposition~\ref{prop:entropy_reg} provides one possible solution for the entropy-regularized unnormalized game~\eqref{eqn:ce_entropy_reg}, other generalized Nash equilibria exist. In particular, each strictly positive generalized Nash equilibrium is an $\epsilon$-correlated equilibrium of the original normal form game~\eqref{eqn:ce_individual_opt}, even if the original normal form game does not have any strictly positive correlated equilibrium. 
\begin{corollary}[$\epsilon$-correlated equilibrium]\label{cor:approx_ce}
If the entropy-regularized generalized Nash equilibrium $(\alpha_1,\ldots,\alpha_N)$ satisfies Assumption~\ref{assum:support invariance} for each $i \in [N]$, their product  $y$~\eqref{eqn:ce_entropy_solution} is an $\epsilon$-correlated equilibria---i.e., for all $i \in [N]$, 
   \begin{equation}\label{eqn:theoretial_bound}
       \sum_{a_{-i}} y(a_i, a_{-i})\Big(\ell_i(a_i, a_{-i}) - \ell_i(\hat{a}_i, a_{-i})\Big) \leq \frac{\epsilon_i}{\lambda_i}, \ \forall a_i, \hat{a}_i \in [A_i], 
   \end{equation}
    where  $\epsilon_i = \max_{a, \hat{a} \in[A]} \log\Big(\alpha_i(a)/\alpha_i(\hat{a})\Big)$ and $\epsilon = \max_i \epsilon_i$. 
   \end{corollary}
\begin{proof}
From~\eqref{eqn:kkt_condition}, we derived the coupled KKT conditions of the entropy-regularized unnormalized game. When $\alpha_i(a) > 0$ for all $(i,a)\in [N]\times[A_i]$, the following holds:
\begin{equation}
    \ell_i(a) + \frac{1}{\lambda_i}\log(\alpha_i(a)) + \sigma_i= 0, \forall a \in [A].
\end{equation}
Then for any $a, {a'} \in [A]$, their cost difference is given by 
\begin{equation}\label{eqn:pf_7}
   \ell_i(a_i, a_{-i}) - \ell_i(a'_i, a_{-i}) = \frac{1}{\lambda_i}\log(\frac{\alpha_i(a_i, a_{-i})}{\alpha_i(a'_i, a_{-i})}).
   \end{equation}
We multiply~\eqref{eqn:pf_7} by $y(a_i, a_{-i})$ and sum over $a_{-i} \in [A_{-i}]$ to obtain
\begin{multline}\label{eqn:pf_8}
   \sum_{a_{-i} \in [A_{-i}]} y(a_i, a_{-i})\Big(\ell_i(a_i, a_{-i}) - \ell_i(a'_i, a_{-i})\Big) = \\
   \frac{1}{\lambda_i}\sum_{a_{-i} \in [A_{-i}]} y(a_i, a_{-i})\log(\frac{\alpha_i(a_i, a_{-i})}{\alpha_i(a'_i, a_{-i})}).
   \end{multline}
Let $\max_{a, {a'}\in[A]} \log(\frac{\alpha_i(a)}{\alpha_i(a')}) = \epsilon_i$. It follows that
\begin{equation}
\sum_{a_{-i} \in [A_{-i}]} y(a_i, a_{-i})\log(\frac{\alpha_i(a_i, a_{-i})}{\alpha_i(a'_i, a_{-i})}) \leq \sum_{a \in [A]} y(a)\epsilon_i \leq \epsilon_i.
\end{equation}
We can then conclude the proposed statement~\eqref{eqn:theoretial_bound}. 
\end{proof}

\section{Computing $\epsilon$-Correlated equilibrium}
We apply the results of Section~\ref{sec:entropy_reg_equilibria} to compute  the generalized Nash equilibrium of the unnormalized game~\eqref{eqn:ce_individual_opt} and evaluate its feasibility as an $\epsilon$-correlated equilibrium of the original normal-form game. 
% We show that using~\eqref{eqn:gNE_entropy_solution}, we are able to approximate efficienctly approximate the correlated equilibrium much faster than using linear program approaches using~\eqref{eqn:ce}. We show this for different combinations of player set $[N]$ and joint action set $[A]$. 

We simulate normal-form games~\eqref{eqn:ne_optimization} with $N = \{2, 3\}$ players and individual action spaces of size $A=\{2,5,10\}$.

Over $K = 1000$ randomly generated normal-form games, we compute the entropy-regularized generalized Nash equilibrium as~\eqref{eqn:gNE_entropy_solution}. We plot the empirical violation with a $5\%$ standard deviation range of the correlated equilibrium condition~\eqref{eqn:ce} under $\epsilon(\text{empirical})$ and the theoretical bound $\max_i\epsilon_i/\lambda_i$~\eqref{eqn:theoretial_bound} under $\epsilon$(bound) for the regularization values $\lambda = \{0.1, 10, 30, 100, 1000, 1e4\}$. We assume that all players use the same entropy regularization, $\lambda_i = \lambda, \ \forall i \in [N]$. 
% The results are grouped by player number and shown in Figure~\ref{fig:qce_approx}. 
\begin{figure*}[ht!]
    \centering
    \includegraphics[width=1.9\columnwidth]{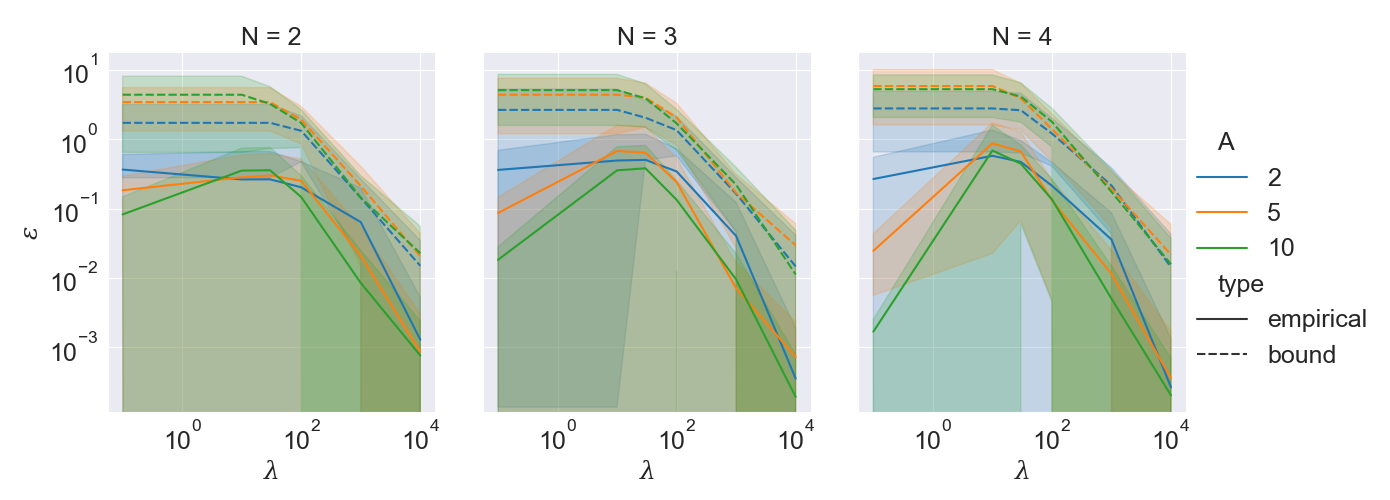}
    \caption{Empirical vs theoretical sub-optimality of the entropy-regularized generalized Nash equilibrium~\eqref{eqn:ce_entropy_solution} as a correlated equilibrium.}
    \label{fig:qce_approx}
\end{figure*}

For each game, we compute its entropy-regularized generalized Nash equilibrium $y^\star$ via~\eqref{eqn:gNE_entropy_solution} and evaluate $y^\star$'s empirical sub-optimality $\epsilon_{ce} = \epsilon$ (empirical) as 
\begin{equation}\label{eqn:empirical_suboptimality}
    \max_{\substack{i \in [N]\\a_i, a'_{i} \in [A_i]}} \sum_{a_{-i}\in [A_{-i}]} y^\star(a_i, a_{-i})\Big(\ell_i(a_i, a_{-i}) - \ell_i(a'_i, a_{-i})\Big).
\end{equation}
We note that $\epsilon_{ce}$ is equivalent to the distance between $y_{ce}$ and the correlated polytope in $\infty$ vector norm.

Finally, we note that a key challenge in applying correlated equilibrium for autonomous interactions is its poor scalability in the number of agents and actions. To this end,~\eqref{eqn:gNE_entropy_solution} provides an approximation that significantly reduces the computation complexity. We observe this in Figure~\ref{fig:qce_time}. 
\begin{figure}[ht!]
    \centering
    \includegraphics[width=\columnwidth]{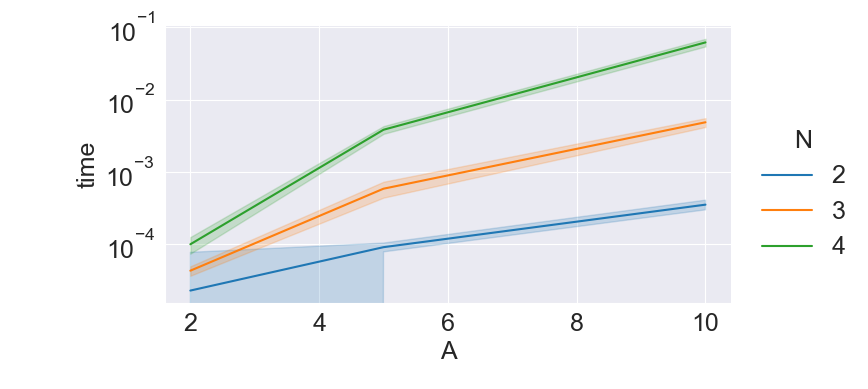}
    \caption{Computation time (seconds) of the $\epsilon$-correlated equilibrium for different numbers of players and actions.}
    \label{fig:qce_time}
\end{figure}

As shown in Figure~\ref{fig:qce_time}, the computation time still scales poorly in the number of actions and players. However, the overall computation time is significantly lower than solving for a feasible point of the correlated equilibrium polytope via linear programming. For comparison, it takes approximately $1.87$ seconds to use CVXPY to compute a correlated equilibrium for the game with $N=3$ players each with $A_i = 3$, whereas approximating it using the entropy-regularized generalized Nash equilibrium only takes $4e-3$ seconds. 

\section{Conclusion}
We introduced an extension of finite player normal-form games to coupled strategies on unnormalized measures over the joint action space and showed that for fully mixed unnormalized measures, the set of generalized Nash equilibria of the unnormalized measure game produces fully mixed correlated equilibria in the original normal-form game. Leveraging the optimization structure this imposes on the correlated equilibria, we introduce an entropy-regularized version of the unnormalized game, and show that its generalized Nash equilibrium is within $\epsilon$ distance of the correlated equilibrium polytope, where $\epsilon$ is dependent on the entropy of each player's unnormalized measure as well as the entropy regularization. Our optimization framework is the first step in connecting correlated equilibrium to the wider literature on gradient-based multi-agent learning algorithms.

\bibliographystyle{IEEEtran}
\bibliography{root}

\end{document}